\documentclass[10pt,a4paper]{article}
\usepackage[utf8]{inputenc}
\usepackage{amsmath, amsfonts, amssymb, graphicx, enumerate, framed, fullpage, calc, subfig, tikz, mathtools, booktabs, todonotes}
\usepackage[amsthm,thmmarks]{ntheorem}
\usepackage[numbers]{natbib}

\usepackage[small]{complexity}
\usepackage{xspace}
\usepackage{colortbl}
\usepackage{placeins}

\makeatletter
\g@addto@macro\@floatboxreset\centering
\makeatother

\newcount\bsubfloatcount
\newtoks\bsubfloattoks
\newdimen\bsubfloatht

\makeatletter
\newcommand{\bsubfloat}[2][]{%
  \sbox\z@{#2}%
  \ifdim\bsubfloatht<\ht\z@
    \bsubfloatht=\ht\z@
  \fi
  \advance\bsubfloatcount\@ne
  \@namedef{bsubfloat\romannumeral\bsubfloatcount}{%
    \subfloat[#1]{\vbox to\bsubfloatht{\hbox{#2}\vfill}}}%
}
\newcommand{\resetbsubfloat}{\bsubfloatcount\z@\bsubfloatht=\z@}
\makeatother

\DeclareMathOperator{\diam}{diam}
\DeclareMathOperator{\td}{td}
\DeclareMathOperator{\tw}{tw}
\DeclareMathOperator{\pw}{pw}
\DeclareMathOperator{\bw}{bw}

\DeclareMathOperator{\rvc}{rvc}
\DeclareMathOperator{\srvc}{srvc}

\newcommand{\vect}[1]{\boldsymbol{\mathbf{#1}}}

\newcommand{\ProblemFormat}[1]{{\sc #1}}
\newcommand{\ProblemName}[1]{\ProblemFormat{#1}\xspace}

\newcommand{\occsat}[0]{\ProblemName{$3$-Occurrence 3-SAT}}
\newcommand{\threesat}[0]{\ProblemName{$3$-SAT}}

\newcommand{\probRc}{\ProblemName{Rainbow Connectivity}}
\newcommand{\probSrc}{\ProblemName{Strong Rainbow Connectivity}}
\newcommand{\probRvc}{\ProblemName{Rainbow Vertex Connectivity}}
\newcommand{\probSrvc}{\ProblemName{Strong Rainbow Vertex Connectivity}}

\newcommand{\probstRvc}{\ProblemName{Rainbow Vertex $st$-Connectivity}}
\newcommand{\probstSrvc}{\ProblemName{Strong Rainbow Vertex $st$-Connectivity}}

\newtheorem{theorem}{Theorem}
\newtheorem{lemma}[theorem]{Lemma}
\newtheorem{corollary}[theorem]{Corollary}

\newtheorem{observation}[theorem]{Observation}

\newenvironment{subproof}[1][\proofname]{%
  \begin{proof}[#1]%
}{%
  \end{proof}%
}

\newcommand{\newres}[0]{\raisebox{0.25ex}{$\bigstar$}}

\colorlet{tableheadcolor}{gray!25} 
\colorlet{tablerowcolor}{gray!20} 
\newcommand{\rowcol}{\rowcolor{tablerowcolor}} %

\title{Complexity of Rainbow Vertex Connectivity\\Problems for Restricted Graph Classes\thanks{Work partially supported by the Emil Aaltonen Foundation}}

\author{Juho Lauri\thanks{Tampere University of Technology, Finland. E-mail: \texttt{juho.lauri@tut.fi}}}

\date{\today}

\begin{document}
\maketitle

\begin{abstract}
A path in a vertex-colored graph $G$ is \emph{vertex rainbow} if all of its internal vertices have a distinct color. The graph $G$ is said to be \emph{rainbow vertex connected} if there is a vertex rainbow path between every pair of its vertices. Similarly, the graph $G$ is \emph{strongly rainbow vertex connected} if there is a shortest path which is vertex rainbow between every pair of its vertices. We consider the complexity of deciding if a given vertex-colored graph is rainbow or strongly rainbow vertex connected. We call these problems \probRvc and \probSrvc, respectively. We prove both problems remain $\NP$-complete on very restricted graph classes including bipartite planar graphs of maximum degree 3, interval graphs, and $k$-regular graphs for $k \geq 3$. We settle precisely the complexity of both problems from the viewpoint of two width parameters: pathwidth and tree-depth. More precisely, we show both problems remain $\NP$-complete for bounded pathwidth graphs, while being fixed-parameter tractable parameterized by tree-depth. Moreover, we show both problems are solvable in polynomial time for block graphs, while \probSrvc is tractable for cactus graphs and split graphs.

\end{abstract}

\section{Introduction}
Krivelevich and Yuster~\citep{Krivelevich2010} introduced the concept of rainbow vertex connectivity. A path in a vertex-colored graph $G$ is said to be \emph{vertex rainbow} if all of its internal vertices have a distinct color. The graph $G$ is said to be \emph{rainbow vertex connected} if there is a vertex rainbow path between every pair of its vertices. The minimum number of colors needed to make $G$ rainbow vertex connected is known as the \emph{rainbow vertex connection number}, and it is denoted by $\rvc(G)$. Recall the \emph{diameter} of a graph $G$, denoted by $\diam(G)$, is the length of a longest shortest path in $G$. It is easy to see two vertices $u$ and $v$ are rainbow vertex connected regardless of the underlying vertex-coloring if their distance $d(u,v)$ is at most~2. Thus, we have that $\rvc(G) \geq \diam(G)-1$, with equality if the diameter is 1 or 2. Similarly, an easy to see upper bound is $\rvc(G) \leq n-2$, as long as we disregard the singleton graph. In other words, complete graphs are precisely the graphs with rainbow vertex connection number 0; for all other graphs we require at least 1 color.

\citet{Li2014} introduced the strong variant of rainbow vertex connectivity. We say the vertex-colored graph $G$ is \emph{strongly rainbow vertex connected} if there is, between every pair of vertices, a shortest path whose internal vertices have a distinct color. The minimum number of colors needed to make $G$ strongly rainbow vertex connected is known as the \emph{strong rainbow vertex connection number}, and it is denoted by $\srvc(G)$. As each strong vertex rainbow coloring is also a rainbow vertex coloring, we have that $\diam(G)-1 \leq \rvc(G) \leq \srvc(G) \leq n-2$.

Prior to the work of Krivelevich and Yuster~\citep{Krivelevich2010}, the concept of rainbow connectivity (for edge-colored graphs) was introduced by~\citet{Chartrand2008} as an interesting way to strengthen the connectivity property. Indeed, the notion has proven to be useful in the domain of networking~\citep{Chakraborty2009} and anonymous communication~\citep{Dorbec2014}. Rainbow coloring and connectivity problems have been subject to considerable interest and research during the past years. For additional applications, we refer the reader to the survey~\citep{Li2012}. A comprehensive introduction is also provided by the books~\citep{Li2012b,Chartrand2008b}.

It is computationally difficult to determine either $\rvc(G)$ or $\srvc(G)$ for a given graph $G$. Indeed, through the work of~\citet{Chen2011} and~\citet{Chen2013} it is known that deciding if $\rvc(G) \leq k$ is $\NP$-complete for every $k \geq 2$. Likewise,~\citet{Eiben2015} showed deciding if $\srvc(G) \leq k$ is $\NP$-complete for every $k \geq 3$. In the same paper, the authors also proved that the strong rainbow vertex connection number of an $n$-vertex graph of bounded diameter cannot be approximated within a factor of $n^{1/2-\epsilon}$, for any $\epsilon > 0$, unless $\P = \NP$. Given such strong intractability results, it is interesting to ask whether the following problem is easier. 
\begin{framed}
\noindent \probRvc (\textsc{Rvc}) \\
\textbf{Instance:} A connected undirected graph $G=(V,E)$, and a vertex-coloring $\psi : V \to C$, where $C$ is a set of colors \\ 
\textbf{Question:} Is $G$ rainbow vertex connected under $\psi$?
\end{framed}
\noindent However, \probRvc was shown to be $\NP$-complete by~\citet{Chen2011}. Later on,~\citet{Huang2014} showed the problem remains $\NP$-complete even when the input graph is a line graph. A more systematic study into the complexity of \probRvc was performed by~\citet{Uchizawa2013}. They proved the problem remains $\NP$-complete for both series-parallel graphs, and graphs of bounded diameter. In contrast, they showed the problem is in $\P$ for outerplanar graphs. Furthermore, they showed the problem is fixed-parameter tractable for the $n$-vertex $m$-edge general graph parameterized by the number of colors in the vertex-coloring. That is, they gave an algorithm running in time $O(k2^kmn)$ such that given a graph vertex-colored with $k$ colors, it decides whether $G$ is rainbow vertex connected.

We mention two related problems, defined on edge-colored undirected graphs. A path in an edge-colored graph $H$ is \emph{rainbow} if no two edges of it are colored the same. The graph $H$ is said to be \emph{rainbow connected} if there is a rainbow path between every pair of its vertices. Likewise, the graph $H$ is said to be \emph{strongly rainbow connected} if there is a shortest path which is rainbow between every pair of its vertices. Formally, the two problems are defined as follows.
\begin{framed}
\noindent \probRc (\textsc{Rc}) \\
\textbf{Instance:} A connected undirected graph $H=(V,E)$, and an edge-coloring $\zeta : E \to C$, where $C$ is a set of colors \\ 
\textbf{Question:} Is $H$ rainbow connected under $\zeta$?
\end{framed}
\begin{framed}
\noindent \probSrc (\textsc{Src}) \\
\textbf{Instance:} A connected undirected graph $H=(V,E)$, and an edge-coloring $\zeta : E \to C$, where $C$ is a set of colors \\ 
\textbf{Question:} Is $H$ strongly rainbow connected under $\zeta$?
\end{framed}
\noindent It was shown by~\citet{Chakraborty2009} that \probRc is $\NP$-complete. Later on, the complexity of both edge variants was studied by~\citet{Uchizawa2013}. For instance, the authors showed both problems remain $\NP$-complete for outerplanar graphs, and that \probRc is $\NP$-complete already on graphs of diameter~2. A further study into the complexity of the edge variant problems was done in our earlier work~\citep{Lauri2015}. For instance, it was shown that both problems remain $\NP$-complete on interval outerplanar graphs, $k$-regular graphs for $k \geq 3$, and on graphs of bounded pathwidth. In addition, block graphs were identified as a class for which the complexity of the two problems \probRc and \probSrc differ. Indeed, it was shown that for block graphs, \probRc is $\NP$-complete, while \probSrc is in $\P$.

In this paper, we introduce as a natural variant of \probRvc the following problem.
\begin{framed}
\noindent \probSrvc (\textsc{Srvc}) \\
\textbf{Instance:} A connected undirected graph $G=(V,E)$, and a vertex-coloring $\psi : V \to C$, where $C$ is a set of colors \\ 
\textbf{Question:} Is $G$ strongly rainbow vertex connected under $\psi$?
\end{framed}
\noindent We present several new complexity results for both \probRvc and \probSrvc.
\begin{itemize}
\item In Section~\ref{sec:hardness_results}, we focus on negative results. In particular, we prove both problems remain $\NP$-complete for bipartite planar graphs of maximum degree~3 (Subsection~\ref{sec:bip}), interval graphs (Subsection~\ref{sec:interval}), triangle-free cubic graphs (Subsection~\ref{sec:cubic}), and $k$-regular graphs for~$k \geq 4$ (Subsection~\ref{sec:regular}). 

\item In Section~\ref{sec:tractability_considerations}, we show both problems are solvable in polynomial time when restricted to the class of block graphs. Furthermore, we extend the algorithm of~\citet{Uchizawa2013} for deciding \probRvc on cactus graphs to decide \probSrvc for the same graph class.

\item In Subsection~\ref{sec:para_consequences}, we consider the implications of our constructions of Section~\ref{sec:hardness_results} for parameterized complexity. For instance, we remark both problems remain $\NP$-complete on graphs of pathwidth $p$, where $p \geq 3$, and also on graphs of bandwidth $b$, where $b \geq 3$. For positive results, we show \probSrvc is $\FPT$ parameterized by the diameter of the input graph, implying polynomial-time solvability for the class of split graphs. Moreover, exploiting known results on tree-depth, we observe all four problems investigated are $\FPT$ parameterized by tree-depth.
\end{itemize}

\section{Preliminaries}
\label{sec:preliminaries}
All graphs we consider in this work are simple, undirected, and finite. We begin by defining the graph classes we consider in this work, along with some terminology and graph invariants. For graph-theoretic concepts not defined here, we refer the reader to~\citep{Diestel2005}. For an integer $n$, we write $[n] = \{1,2,\ldots,n\}$.

A \emph{coloring} of a graph $G$ is an assignment of colors to the vertices of $G$ such that no two adjacent vertices receive the same color. A graph $G$ is said to be \emph{$k$-colorable} if there exists a coloring using $k$ colors for it. A 2-colorable graph is \emph{bipartite}. A \emph{complete graph} on $n$ vertices, denoted by $K_n$, has all the possible ${n \choose 2}$ edges. In particular, we will call $K_3$ a \emph{triangle}. A \emph{complete bipartite graph} consists of two non-empty independent sets $X$ and $Y$ with $(x,y)$ being an edge whenever $x \in X$ and $y \in Y$. A complete bipartite graph is denoted by $K_{n,m}$, and it has $n+m = |X|+|Y|$ vertices. In particular, we will call $K_{1,3}$ a \emph{claw}. A complete subgraph of $G$ is a \emph{clique}. The \emph{clique number} of a graph $G$, denoted by $\omega(G)$, is the size of a largest clique in $G$.

A graph is said to be \emph{planar} if it can be embedded in the plane with no crossing edges. Equivalently, a graph is planar if it is $(K_{3,3},K_5)$-minor-free. A graph is \emph{outerplanar} if it has a crossing-free embedding in the plane such that all vertices are on the same face. Clearly, each outerplanar graph is planar. Another superclass of outerplanar graphs is formed by \emph{series-parallel graphs}. Series-parallel graphs are exactly the $K_4$-minor-free graphs~\citep{Duffin1965}. In a \emph{cactus graph}, every edge is in at most one cycle. Cactus graphs form a subclass of outerplanar graphs.

A \emph{chord} is an edge joining two non-consecutive vertices in a cycle. A graph is \emph{chordal} if every cycle of length 4 or more has a chord. Equivalently, a graph is chordal if it contains no induced cycle of length 4 or more. Chordal graphs are precisely the class of graphs admitting a \emph{clique tree}~\citep{Gavril1974}. A clique tree of a connected chordal graph $G$ is any tree $T$ whose vertices are the maximal cliques of $G$ such that for every two maximal cliques $C_i,C_j$, each clique on the path from $C_i$ to $C_j$ in $T$ contains $C_i \cap C_j$. A subclass of chordal graphs is formed by \emph{interval graphs}. A graph is an interval graph if and only if it admits a clique tree that is path~\citep{Gilmore1964}. A \emph{cut vertex} is a vertex whose removal will disconnect the graph. A \emph{biconnected graph} is a connected graph with no cut vertices. In a \emph{block graph}, every maximal biconnected component, known as a \emph{block}, is a clique. In other words, every edge of a block graph $G$ lies in a unique block, and $G$ is the union of its blocks. It is easy to see that block graphs are also chordal. An \emph{independent set} in a graph is a set of pairwise non-adjacent vertices. A graph whose vertex set can be partitioned into a clique and an independent set is known as a \emph{split graph}. It is easy to see that a split graph is chordal.

The \emph{degree} of a vertex $v$ is the number of edges incident to $v$. A graph is \emph{$k$-regular} if every vertex has degree exactly $k$. In particular, we will call a 3-regular graph \emph{cubic}. A connected 2-regular graph is a \emph{cycle graph}. A cycle graph on $n$ vertices is denoted by $C_n$.

A \emph{proper interval graph} is a graph that is both interval and claw-free (see~\citep{Roberts1969}). The \emph{bandwidth} of a graph~$G$, denoted by $\bw(G)$, is one less than the minimum clique number of any proper interval graph having $G$ as a subgraph~\citep{Kaplan1996}. The \emph{pathwidth} of a graph~$G$, denoted by $\pw(G)$, is one less than the minimum clique number of any interval graph having $G$ as a subgraph. The \emph{treewidth} of a graph~$G$, denoted by $\tw(G)$, is one less than the minimum clique number of any chordal graph having $G$ as a subgraph. Indeed, for a graph $G$, we have that $\tw(G) \leq \pw(G) \leq \bw(G)$ (for a proof, see~\citep{Bodlaender1998}). Finally, a $(C_4,P_4)$-free graph is \emph{trivially perfect}. The \emph{tree-depth} of a graph $G$, denoted by $\td(G)$, is the minimum clique number of any trivially perfect graph having $G$ as a subgraph. Here, we have that $\pw(G) \leq \td(G)-1$ (for a proof, see~\citep{Bodlaender1995}).

Finally, we say a problem is \emph{fixed-parameter tractable} (FPT) if it can be solved in time $f(k) \cdot n^{O(1)}$, where $f$ is some computable function depending solely on some parameter $k$, and $n$ is the input size. Similarly, a problem is said to be in $\XP$ if it can be solved in $n^{f(k)}$ time. For a more comprehensive treatment on parameterized complexity, we refer the reader to the books~\citep{Downey2013,Cygan2015}.

\section{Hardness results}
\label{sec:hardness_results}
In this section, we will give a number of hardness results for both \probRvc and \probSrvc for very restricted graph classes. It is interesting to compare the obtained complexity results against those of the edge variants, namely \probRc and \probSrc. Indeed, we summarize the known complexity results for all four variants in Table~\ref{tbl:hardness_summary} along with our new results.

\begin{table}[t]
\caption{Complexity results for rainbow connectivity problems along with some of our new results marked by $\bigstar$. The symbol $\dagger$ stands for~\cite{Uchizawa2013} and the symbol $\ddagger$ for~\cite{Lauri2015}.}
\label{tbl:hardness_summary}
\centering
\def\arraystretch{1}
\setlength{\tabcolsep}{10pt}
\begin{tabular}{lllll}
\toprule 
Graph class 				& \textsc{Rvc} & \textsc{Srvc} & \textsc{Rc} & \textsc{Src} \\ 
\midrule
Block 						& $\P$ \newres & $\P$ \newres & $\NPC$ $\ddagger$ & $\P$ $\ddagger$ \\
Bounded bandwidth & $\NPC$ \newres & $\NPC$ \newres & $\NPC$ $\ddagger$ & $\NPC$ $\ddagger$ \\
Bounded diameter			& $\NPC$ $\dagger$ & $\FPT$ \newres & $\NPC$ $\dagger$ & $\FPT$ \newres \\ 
\rowcol Bounded pathwidth & $\NPC$ \newres & $\NPC$ \newres & $\NPC$ $\ddagger$ & $\NPC$ $\ddagger$ \\
\rowcol Bounded tree-depth & $\FPT$ \newres & $\FPT$ \newres & $\FPT$ \newres & $\FPT$ \newres \\
\rowcol Cactus 						& $\P$ $\dagger$ & $\P$ \newres & $\P$ $\dagger$ & $\P$ $\dagger$ \\
Interval					& $\NPC$ \newres & $\NPC$ \newres & $\NPC$ $\ddagger$ & $\NPC$ $\ddagger$ \\
$k$-regular, $k \geq 3$		& $\NPC$ \newres & $\NPC$ \newres & $\NPC$ $\ddagger$ & $\NPC$ $\ddagger$ \\
Outerplanar 				& $\P$ $\dagger$ & ? & $\NPC$ $\dagger$ & $\NPC$ $\dagger$ \\
\rowcol Series-parallel 			& $\NPC$ $\dagger$ & $\NPC$ \newres & $\NPC$ $\dagger$ & $\NPC$ $\dagger$ \\
\rowcol Split						& ? & $\P$ \newres & ? & $\P$ \newres \\
\rowcol Tree 						& $\P$ & $\P$ & $\P$ & $\P$ \\
\bottomrule 
\end{tabular}
\end{table}

\subsection{Overview of the reductions}
\label{sec:overview}
In this subsection, we give an overview of our reductions. Let us remark that all of the four problems considered are in $\NP$ with the certificate being a set of colored paths, one path for each pair of vertices. All of our reductions are from the \occsat problem, which is a variant of the classical \threesat problem. In the \occsat problem, we have a restriction that every variable occurs at most three times, and each clause has at most 3 literals. The problem is known to be $\NP$-complete~\citep{Papadimitriou1994}.

All of our reductions are greatly inspired by those of~\citet{Uchizawa2013}, who gave hardness results for both \probRc and \probRvc. Let us explain the gist of their reduction on a high-level. Given a \occsat formula $\phi$, a variable gadget is constructed for each variable, and a clause gadget is built for each clause. Moreover, a certain vertex-coloring is constructed for each gadget. A key idea is that regardless of the satisfiability of $\phi$, every vertex pair in a gadget is rainbow (vertex) connected. Moreover, regardless of the satisfiability of $\phi$, the whole graph will be rainbow (vertex) connected except for a specific vertex pair $s$ and $t$. Informally, the gadgets are set up in a path-like manner, and the special vertices $s$ and $t$ act as endpoints of this path-like graph. The idea is illustrated in Figure~\ref{fig:planar_bip_construction} (the corresponding construction is given in Theorem~\ref{thm:rvc_bip_planar}).

Clearly, a strongly rainbow (vertex) connected graph is also rainbow (vertex) connected. Therefore, it is desirable to construct the gadgets such that each vertex pair is always (regardless of $\phi$) connected by a rainbow (vertex) shortest path. This allows one to obtain a hardness result for the strong problem variant as well. Indeed, the first hardness result we present (Theorem~\ref{thm:rvc_bip_planar}) will be of this flavor. However, we are not always able to do this, or doing so will overly complicate the construction in question. We will always explicitly mark whether or not this is the case, i.e., if a hardness result for the strong variant follows as well.

Finally, for the sake of presentation, all of our constructions assume the given \occsat formula $\phi$ only has clauses with exactly 3 literals. However, as shown by~\citet{Tovey1984}, every such instance is satisfiable. Therefore, we will present clause gadgets corresponding to clauses of size~2 in the appendix for each graph class (note that clauses of size~1 can be safely removed by unit propagation).

\subsection{Bipartite planar graphs}
\label{sec:bip}
In this subsection, we will prove that both \probRvc and \probSrvc remain $\NP$-complete on bipartite planar graphs of maximum degree 3. We remark that this is a very restricted graph class, generalizing the class of bipartite claw-free graphs. A bipartite claw-free graph consists of disjoint cycles and paths. It is easy to see both \probRvc and \probSrc are solvable in polynomial time for the class of bipartite claw-free graphs.
 
\begin{theorem}
\label{thm:rvc_bip_planar}
\probRvc is $\NP$-complete when restricted to the class of bipartite planar graphs of maximum degree 3.
\end{theorem}
\textbf{Construction}: Given a \occsat formula $\phi = \bigwedge_{j=1}^{m} c_i$ over variables $x_1,x_2,\ldots,x_n$, we construct a graph $G_\phi$ and a vertex-coloring $\psi$ such that $\phi$ is satisfiable if and only if $G_\phi$ is rainbow vertex connected under $\psi$. We first describe the construction of $G_\phi$, and then the vertex-coloring $\psi$ of $G_\phi$.

\begin{figure}
\bsubfloat[]{%
  \includegraphics[scale=1]{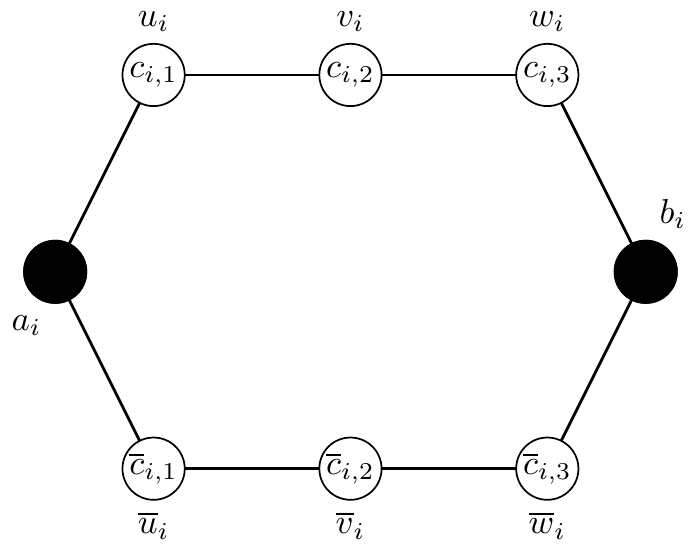}%
}
\bsubfloat[]{%
  \includegraphics[scale=1]{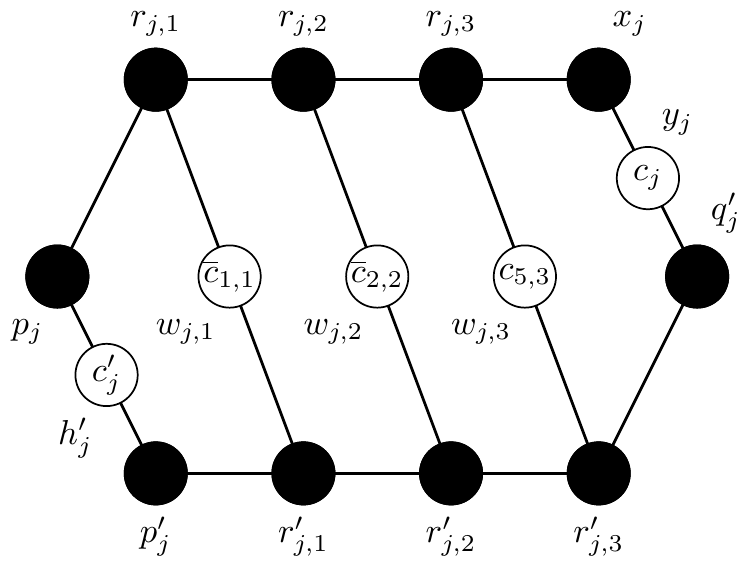}%
}
\bsubfloati\qquad\bsubfloatii
\caption{\textbf{(a)} A variable gadget $X_i$ for the variable $x_i$, and \textbf{(b)} a clause gadget $C_j$ for the clause $c_j = (x_1 \vee x_2 \vee \neg x_5)$, where $x_1$ is the first literal of $x_1$, $x_2$ is the second literal of $x_2$, and $\neg x_5$ is the third literal of $x_5$. The vertices receiving fresh distinct colors are drawn as solid circles.}
\label{fig:planar_bip_gadgets}
\end{figure}

We will construct for each variable $x_i$, where $i \in [n]$, a \emph{variable gagdet} $X_i$. A variable gadget $X_i$ is the cycle graph $C_8$ embedded in the plane on the vertices $a_i$, $u_i$, $v_i$, $w_i$, $b_i$, $\overline{w}_i$, $\overline{v}_i$, $\overline{u}_i$ in clockwise order. For each clause $c_j$, where $j \in [m]$, we construct a \emph{clause gadget} $C_j$. A clause gadget $C_j$ is built by starting from the cycle graph $C_{12}$ embedded in the plane on the vertices $p_j$, $r_{j,1}$, $r_{j,2}$, $r_{j,3}$, $x_j$, $y_j$, $q'_j$, $r'_{j,3}$, $r'_{j,2}$, $r'_{j,1}$, $p'_j$, and $h'_j$ in clockwise order, and by adding chords $(r_{j,1},r'_{j,1})$, $(r_{j,2},r'_{j,2})$, and $(r_{j,3},r'_{j,3})$. For $\ell \in [3]$, the added chord $(r_{j,\ell},r'_{j,\ell})$ is subdivided by a new vertex $w_{j,\ell}$. The vertices $w_{j,\ell}$ correspond to the three literals the clause $c_j$ has. Both a variable gadget and a clause gadget are shown in Figure~\ref{fig:planar_bip_gadgets}.


For each $1 \leq i < n$, we connect $X_i$ with $X_{i+1}$ by adding a new vertex $d_i$ along with two edges $(b_i,d_i)$ and $(d_i,a_{i+1})$. Similarly, we connect $C_j$ with $C_{j+1}$ by adding a new vertex $f_j$ along with two edges $(q'_j,f_j)$ and $(f_j,p_{j+1})$ for each $1 \leq j < m$. The two components are connected together by adding the vertex $d_n$ with the edges $(b_n,d_n)$ and $(d_n,p_1)$. We then add two vertices $t'$ and $t$ along with the edges $(q'_m,t')$ and $(t',t)$. Finally, we construct a path of length $m+1$ on vertices $s_0,s_1,\ldots,s_m$, and connect it with $G_\phi$ by adding the edge $(s_m,a_1)$. This completes the construction of $G_\phi$. We can verify $G_\phi$ is indeed a bipartite planar graph of maximum degree 3.


We then describe the vertex-coloring $\psi$ given to the vertices of $G_\phi$. Observe that in a variable gadget $X_i$, there are precisely two paths between $a_i$ and $b_i$. Intuitively, taking the path from $a_i$ to $b_i$ through $u_i$, $v_i$, and $w_i$ corresponds to setting $x_i = 1$ in the formula $\phi$; we refer to this path as the \emph{positive $X_i$ path}. We color the three vertices $u_i$, $v_i$, and $w_i$ with colors $c_{i,1}, c_{i,2}$, and $c_{i,3}$, respectively. Taking the path from $a_i$ to $b_i$ through $\overline{u}_i$, $\overline{v}_i$, and $\overline{w}_i$ corresponds to setting $x_i = 0$ in the formula $\phi$; we refer to this path as the \emph{negative $X_i$ path}. The three vertices $\overline{u}_i$, $\overline{v}_i$, and $\overline{w}_i$ receive colors $\overline{c}_{i,1},\overline{c}_{i,2}$ and $\overline{c}_{i,3}$, respectively. The coloring of a variable gadget $X_i$ is illustrated in Figure~\ref{fig:planar_bip_gadgets} (a).

\begin{figure}[t]
\includegraphics[scale=0.74]{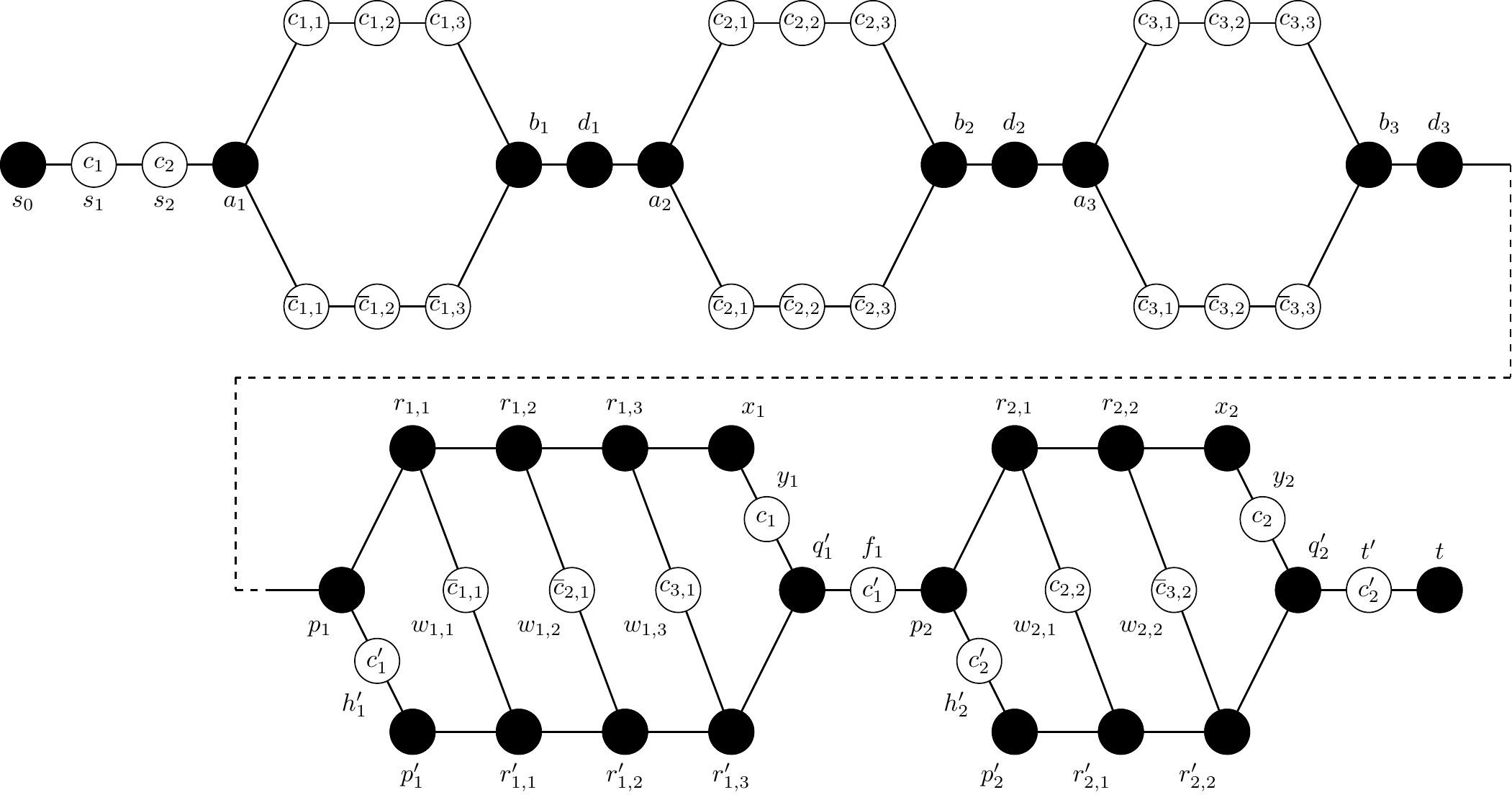}
\caption{A planar bipartite graph $G_\phi$ of maximum degree~3 constructed for the formula $\phi = (x_1 \vee x_2 \vee \neg x_3) \wedge (\neg x_2 \vee x_3)$. For brevity, some vertex labels are not shown.}
\label{fig:planar_bip_construction}
\end{figure}


Recall that a variable $x_i$ appears at most three times in $\phi$. We refer to the first occurrence of $x_i$ as the \emph{first literal of $x_i$}, the second occurrence of $x_i$ as the \emph{second literal of $x_i$}, and finally the third occurrence of $x_i$ as the \emph{third literal of $x_i$}. If a clause has two or three literals of a same variable, the tie is broken arbitrarily. In a clause gadget $C_j$, we color vertex $h'_j$ with color $c_j'$, and vertex $y_j$ with color $c_j$. For each $k \in [3]$, we denote the $k$th literal in the $j$th clause by $l_{j,k}$. We color vertex $w_{j,\ell}$ as follows:
\[
 \psi(w_{j,\ell}) = 
 \begin{dcases*}
        \overline{c}_{i,1} 	& if $l_{j,k}$ is a positive literal and the first literal of $x_i$ \\
        \overline{c}_{i,2} 	& if $l_{j,k}$ is a positive literal and the second literal of $x_i$ \\
        \overline{c}_{i,3} 	& if $l_{j,k}$ is a positive literal and the third literal of $x_i$ \\
        c_{i,1} 		& if $l_{j,k}$ is a negative literal and the first literal of $x_i$ \\
        c_{i,2} 		& if $l_{j,k}$ is a negative literal and the second literal of $x_i$ \\
        c_{i,3} 		& if $l_{j,k}$ is a negative literal and the third literal of $x_i$
  \end{dcases*}
\]
The vertex $f_j$, for each $1 \leq j < m$, receives color $c_j'$, while vertex $t'$ is colored with $c'_m$. The coloring of a clause gadget $C_j$ is shown in Figure~\ref{fig:planar_bip_gadgets} (b).


Finally, for each $1 \leq j \leq m$, we color vertex $s_j$ with color $c_j$. Every other uncolored vertex of $G_\phi$ receives a fresh new color that does not appear in $G_\phi$. Formally, these are precisely the vertices in 
\begin{equation*}
\begin{split}
U &= \{ a_i,b_i,d_i \mid 1 \leq i \leq n \} \\
\quad &\cup \{ p_j, r_{j,1}, r_{j,2}, r_{j,3}, x_j, q'_j, r'_{j,3}, r'_{j,2}, r'_{j,1}, p'_j \mid 1 \leq j \leq m \} \\
\quad &\cup \{ s_0,t \}.
\end{split}
\end{equation*}
Vertices in $U$ shown in Figure~\ref{fig:planar_bip_gadgets} are drawn as solid circles. This completes the vertex-coloring $\psi$ of $G_\phi$. An example is shown in Figure~\ref{fig:planar_bip_construction}. The proof of Theorem~\ref{thm:rvc_bip_planar} is obtained via the following two lemmas, which also make precise the intuition provided in Section~\ref{sec:overview}. The arguments essentially follow from~\citep{Uchizawa2013}, but we describe them for completeness. The reader should observe the two following lemmas prove a slightly stronger statement than necessary, by talking about \emph{strong} rainbow vertex connectedness instead of rainbow vertex connectedness.

\begin{lemma}
\label{lem:g_srvc_iff_st_path}
The graph $G_\phi$ is strongly rainbow vertex connected under the vertex-coloring $\psi$ if and only if $G_\phi$ has a vertex rainbow shortest path between the vertices $s_0$ and $t$.
\end{lemma}
\begin{proof}
Trivially, it suffices show that if $s_0$ and $t$ are strongly rainbow vertex connected, then $G_\phi$ is strongly rainbow vertex connected. For convenience, we partition the vertex set $V$ into three groups. Indeed, let $V = S \cup A \cup L$, where $S = \{s_1,\ldots,s_m\}$, $A = \bigcup_{i=1}^{n} V(X_i)$, and $L = \bigcup_{j=1}^{m} V(C_j)$. Let $u$ and $v$ be two distinct vertices in $V$, and we will show they are strongly rainbow vertex connected. It is straightforward to verify $u$ and $v$ are strongly rainbow vertex connected when they are in the same group. So let us consider the three possible cases of $u$ and $v$ being in distinct groups.

\begin{itemize}
\item \textbf{Case 1:} $u \in S$ and $v \in A$ are strongly rainbow vertex connected.
\begin{subproof}
No two vertices in $S$ and $A$ share colors, so the claim follows.
\end{subproof}

\item \textbf{Case 2:} $u \in S$ and $v \in L$ are strongly rainbow vertex connected.
\begin{subproof}
By our assumption, there is a rainbow shortest path $P$ from $s_0$ and $t$. Observe that $P$ must use every color $c_1,\ldots,c_m$, and also every color $c'_1,\ldots,c'_m$. Therefore, it must be the case that $P$ uses the vertex $w_{j,\ell}$ for some $\ell \in [3]$ for every $j \in [m]$. So suppose the vertex $v$ is contained in a clause gadget $C_j$. By the above reasoning, it is clear that $p_j$ is reachable from $u \in S$ by a rainbow shortest path $P'$, which is a subpath of $P$. Finally, we can construct a shortest path $P''$ from $p_j$ to $v$ such that $y_j$ is not an internal vertex of $P''$. The concatenation of $P'$ and $P''$ gives us a rainbow shortest path between $u$ and $v$, so the claim follows.
\end{subproof}

\item \textbf{Case 3:} $u \in A$ and $v \in L$ are strongly rainbow vertex connected.
\begin{subproof}
Observe that we can always choose a shortest $u$-$v$ path $P$ so that none of the vertices $w_{j,\ell}$ appear as an internal vertex in $P$, for any $j \in [m]$ and $\ell \in [3]$. Thus, the claim follows.
\end{subproof}
\end{itemize}
This completes the proof.
\end{proof}

\begin{lemma}
\label{lem:g_rainbow_conn}
There is a vertex rainbow shortest path between $s_0$ and $t$ if and only if the formula $\phi$ is satisfiable.
\end{lemma}
\begin{proof}
Suppose there is a vertex rainbow shortest path $P$ between $s_0$ and $t$, and we will show the formula $\phi$ is satisfiable. It is clear that $P$ must choose from every variable gadget $X_i$ either the positive or the negative $X_i$ path. Indeed, let us construct a truth assignment $\vect{\alpha} =(\alpha_1,\ldots,\alpha_n)$ for $\phi$ as follows. For every $X_i$, if $P$ is using the positive $X_i$ path, we set $\alpha_i = 1$. Otherwise, $P$ is using the negative $X_i$ path and we let $\alpha_i = 0$. We will then argue $\vect{\alpha}$ is a satisfying assignment for $\phi$. Consider a clause gadget $C_j$, where $j \in [m]$. It is easy to verify the vertex rainbow shortest path $P$ must use exactly one of the vertices $w_{j,\ell}$, where $\ell \in [3]$, in every $C_j$. Indeed, if two or more of the vertices $w_{j,\ell}$ were chosen, the path $P$ would not be a shortest path. So consider the vertex $w_{j,\ell}$ chosen by $P$ in some clause gadget $C_j$. Furthermore, suppose $w_{j,\ell}$ has received color $c_{i,\delta}$, for some $i \in [n]$ and $\delta \in [3]$ (recall a variable occurs at most three times in $\phi$). By construction, the literal $l_{j,k}$ corresponding to $w_{j,\ell}$ is a negative literal of the variable $x_i$. Moreover, color $c_{i,\delta}$ also appears on the positive $X_i$ path. Because $P$ contains $w_{j,\ell}$ colored $c_{i,\delta}$, it follows $P$ chooses the $X_i$ negative path. Thus, we have $\alpha_i = 0$, and the literal $l_{j,k}$ is set true by~$\vect{\alpha}$. The proof is symmetric for the case $w_{j,\ell}$ having color $\overline{c}_{i,\delta}$.

For the other direction, suppose $\phi$ is satisfiable under the assignment $\vect{\alpha} = (\alpha_1,\ldots,\alpha_n)$. We construct a vertex rainbow shortest path $P$ between $s_0$ and $t$ as the concatenation of two paths $P_V$ and $P_C$. To construct $P_V$, we proceed as follows. For each variable gadget $X_i$, if $\alpha_i = 1$ we choose the positive $X_i$ path; otherwise $\alpha_i = 0$ and we choose the $X_i$ negative path. Clearly, $P_V$ is a vertex rainbow shortest path from $s_0$ to $b_n$. We will then show that for every clause gadget $C_j$, there is a vertex $w_{j,\ell}$ such that its color does not appear on $P_V$. It will then be straightforward to construct the path $P_C$. Because $\vect{\alpha}$ is a satisfying assignment for $\phi$, each clause has a literal which is made true by $\vect{\alpha}$. Let $l_{j,k}$ be such a literal for a clause gadget $C_j$. Suppose $l_{j,k}$ is a positive literal of the variable $x_i$, for some $i \in [n]$. By construction, the vertex $w_{j,l}$ has received color $\overline{c}_{i,\delta}$, where $\delta \in [3]$. Because $l_{j,k}$ is a positive literal of $x_i$ and $l_{j,k}$ is made true by $\vect{\alpha}$, we have that $\alpha_i = 1$. Moreover, the path $P_V$ has taken the positive $X_i$ path, meaning it is using color $c_{i,1}$, $c_{i,2}$, and $c_{i,3}$. In other words, color $\overline{c}_{i,\delta}$ does not appear in $P_V$. Thus, the concatenation of $P_V$ and $P_C$ indeed gives us a vertex rainbow shortest path between $s_0$ and $t$. This completes the proof.
\end{proof}
For proving Theorem~\ref{thm:rvc_bip_planar}, the two above lemmas are slightly stronger than necessary. That is, given a positive instance of $\phi$, every pair of vertices in $G_\phi$ is not only rainbow vertex connected, but \emph{strongly} rainbow vertex connected. In other words, we have also proven the following.

\begin{theorem}
\probSrvc is $\NP$-complete when restricted to the class of bipartite planar graphs of maximum degree 3.
\end{theorem}

\subsection{Interval graphs}
\label{sec:interval}
In this subsection, we investigate the complexity of both \probRvc and \probSrvc on chordal graphs. We will show that both problems remain $\NP$-complete on interval graphs, which form a well-known subclass of chordal graphs. In fact, we will prove a stronger result for \probSrvc by showing it remains $\NP$-complete for proper interval graphs. A \emph{caterpillar} is a tree that has a dominating path. One can observe caterpillars form a subclass of interval graphs. Moreover, both problems are solvable in polynomial time on caterpillars.

\begin{theorem}
\label{thm:rvc_interval}
\probRvc is $\NP$-complete when restricted to the class of interval graphs.
\end{theorem}
\begin{proof}
We assume the terminology of Theorem~\ref{thm:rvc_bip_planar}. Given a \occsat instance $\phi = \bigwedge_{j=1}^{m} c_i$ over variables $x_1,x_2,\ldots,x_n$, we follow a strategy similar to Theorem~\ref{thm:rvc_bip_planar}. We will first describe how variable and clause gadgets of a graph $G^I_\phi$ are built along with their vertex-colorings.

\begin{figure}
\includegraphics[scale=1]{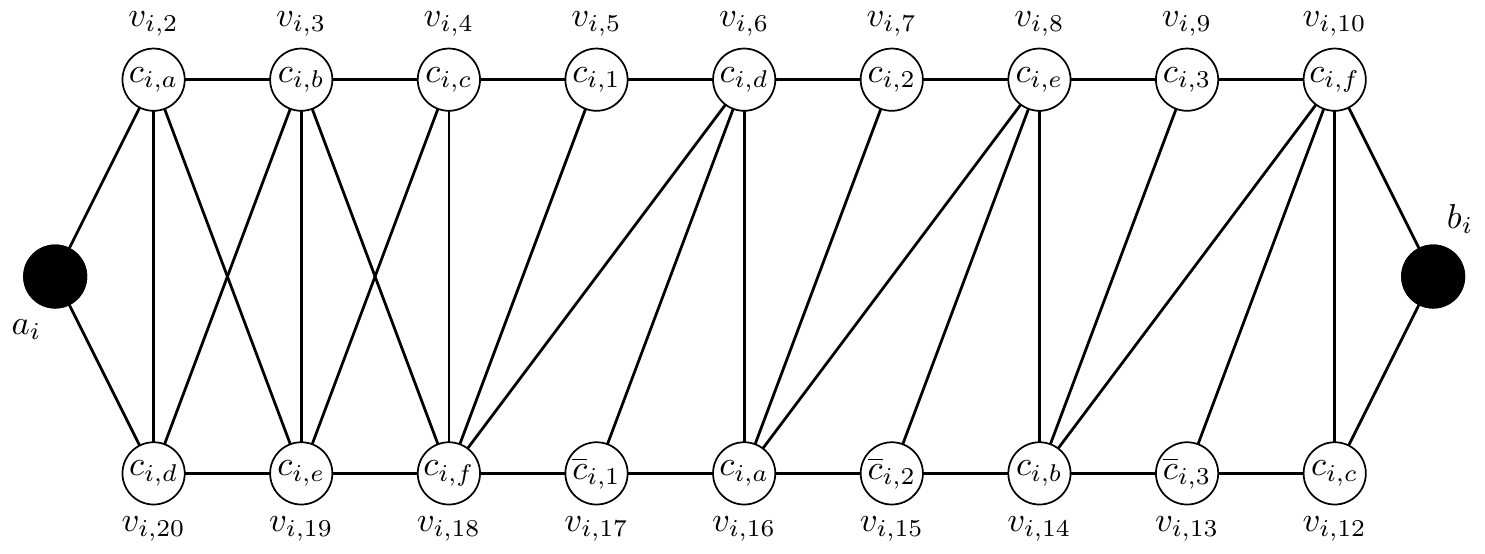}
\caption{A variable gadget $X^I_i$.}
\label{fig:chordal_variable_gadget}
\end{figure}

A variable gadget $X^I_i$ is built by starting from the cycle graph $C_{20}$ on the vertices $v_{i,\ell}$ in clockwise order, where $\ell \in [20]$. For convenience (and to match Theorem~\ref{thm:rvc_bip_planar}), we rename $v_{i,1}$ to $a_i$ and $v_{i,11}$ to $b_i$. We will then describe the altogether 19 chords added to $X^I_i$. First, we add the chords $(v_{i,2},v_{i,19})$, $(v_{i,2},v_{i,20})$, $(v_{i,3},v_{i,18})$, $(v_{i,3},v_{i,19})$, $(v_{i,3},v_{i,20})$, $(v_{i,4},v_{i,18})$, and $(v_{i,4},v_{i,19})$. Then, we add the chords $(v_{i,5},v_{i,18})$, $(v_{i,6},v_{i,16})$, $(v_{i,6},v_{i,17})$, $(v_{i,6},v_{i,18})$, $(v_{i,7},v_{i,16})$, $(v_{i,8},v_{i,14})$, $(v_{i,8},v_{i,15})$, $(v_{i,8},v_{i,16})$, $(v_{i,9},v_{i,14})$, $(v_{i,10},v_{i,14})$, $(v_{i,10},v_{i,13})$, and $(v_{i,10},v_{i,12})$. This completes the construction of a variable gadget $X^I_i$. A variable gadget $X^I_i$ is shown in Figure~\ref{fig:chordal_variable_gadget}. It is straightforward to verify $X^I_i$ admits a clique tree that is a path, and thus $X^I_i$ is an interval graph. 

We will then describe the vertex-coloring of $X^I_i$. For each $X^I_i$, we introduce 6 new colors $c_{i,a}$, $c_{i,b}$, $c_{i,c}$, $c_{i,d}$, $c_{i,e}$, and $c_{i,f}$. We color both vertices $v_{i,2}$ and $v_{i,16}$ with color $c_{i,a}$, both $v_{i,3}$ and $v_{i,14}$ with color $c_{i,b}$, both $v_{i,4}$ and $v_{i,12}$ with color $c_{i,c}$, both $v_{i,20}$ and $v_{i,6}$ with color $c_{i,d}$, both $v_{i,19}$ and $v_{i,8}$ with color $c_{i,e}$, and both $v_{i,18}$ and $v_{i,10}$ with color $c_{i,f}$. The vertices $v_{i,5}$, $v_{i,7}$, and $v_{i,9}$ receive colors $c_{i,1}$, $c_{i,2}$, and $c_{i,3}$, respectively. Similarly, the vertices $v_{i,17}$, $v_{i,15}$, and $v_{i,13}$ receive colors $\overline{c}_{i,1}$, $\overline{c}_{i,2}$, and $\overline{c}_{i,3}$, respectively. Conceptually, these two sets of three vertices correspond to the positive and the negative $X_i$ path of Theorem~\ref{thm:rvc_bip_planar}. The vertex-coloring of a variable gadget $X^I_i$ is shown in Figure~\ref{fig:chordal_variable_gadget}.

A clause gadget $C^I_j$ is built by starting from a clause gadget $C_j$, and by adding the altogether 15 chords $(r_{j,1},h'_j)$, $(r_{j,1},p'_j)$, $(r_{j,2},p'_j)$, $(r_{j,2},w_{j,1})$, $(r_{j,2},r'_{j,1})$, $(r_{j,3},r'_{j,1})$, $(r_{j,3},w_{j,2})$, $(r_{j,3},r'_{j,2})$, $(x_j,r'_{j,2})$, $(x_j,w_{j,3})$, $(x_j,r'_{j,3})$, $(y_j,r'_{j,3})$, $(p'_j,w_{j,1})$, $(r'_{j,1},w_{j,2})$, and $(r'_{j,2},w_{j,3})$. This completes the construction of a clause gadget $C^I_j$. A clause gadget $C^I_j$ is shown in Figure~\ref{fig:chordal_clause_gadget}. It can be verified $C^I_j$ admits a clique tree that is a path, and thus $C^I_j$ is an interval graph.

We will then describe the vertex-coloring of $C^I_j$. We color vertices $h'_j$ and $y_j$, and the three vertices $w_{j,\ell}$, for $\ell \in [3]$, exactly as in Theorem~\ref{thm:rvc_bip_planar}. Moreover, for each $C^I_j$, we introduce four new colors $c_{j,u}$, $c_{j,v}$, $c_{j,w}$, and $c_{j,z}$. We color both vertices $r_{j,1}$ and $p'_j$ with color $c_{j,u}$, both $r_{j,2}$ and $r'_{j,1}$ with color $c_{j,v}$, both $r_{j,3}$ and $r'_{j,2}$ with color $c_{j,w}$, and both $x_j$ and $r'_{j,3}$ with color $c_{j,z}$. The vertex-coloring of a clause gadget $C^I_j$ is shown in Figure~\ref{fig:chordal_clause_gadget}.

The variable and clause gadgets are joined together precisely as in Theorem~\ref{thm:rvc_bip_planar}. Furthermore, we also add vertices $s,s_1,\ldots,s_m$, $t'$, and $t$ exactly as in Theorem~\ref{thm:rvc_bip_planar}. The remaining uncolored vertices receive a fresh new color that does not appear in $G^I_\phi$. Formally, these are precisely the vertices in
\begin{equation*}
\begin{split}
U &= \{ a_i,b_i,d_i \mid 1 \leq i \leq n \} \\
\quad &\cup \{ p_j, q'_j \mid 1 \leq j \leq m \} \\
\quad &\cup \{ s_0,t \}.
\end{split}
\end{equation*}
Informally, disregarding the vertex-colorings, the graph $G^I_\phi$ differs from the graph $G_\phi$ of Theorem~\ref{thm:rvc_bip_planar} only in the way in which the gadgets are built.

We will then show these modifications do not contradict Lemma~\ref{lem:g_rainbow_conn}. Since we only modified the variable and clause gadgets, it suffices to inspect them. Consider a variable gadget $X^I_i$. We claim that any vertex rainbow path $R$ from $a_i$ to $b_i$ must still pass through all the vertices in either $P = \{ v_{i,5}, v_{i,7}, v_{i,9} \}$ or $N = \{ v_{i,17}, v_{i,15}, v_{i,13} \}$. Observe that the path $R$ must choose at least one vertex from each set of two vertices at a distance 1, 2, and 3 from $a_i$. Similarly, by the way in which the vertices are colored, the path $R$ must choose exactly one vertex from the set of two vertices at a distance 5, 7, and 9 from $a_i$. Thus, $R$ cannot choose more than three vertices from $\{ v_{i,\ell}, v_{i,16+\ell} \mid 2 \leq \ell \leq 4 \}$. It is then straightforward to verify that $R$ must pass through all the vertices in either $P$ or $N$. In other words, there are exactly two choices how the path $R$ can traverse from $a_i$ to $b_i$.

\begin{figure}[t]
\includegraphics[scale=1]{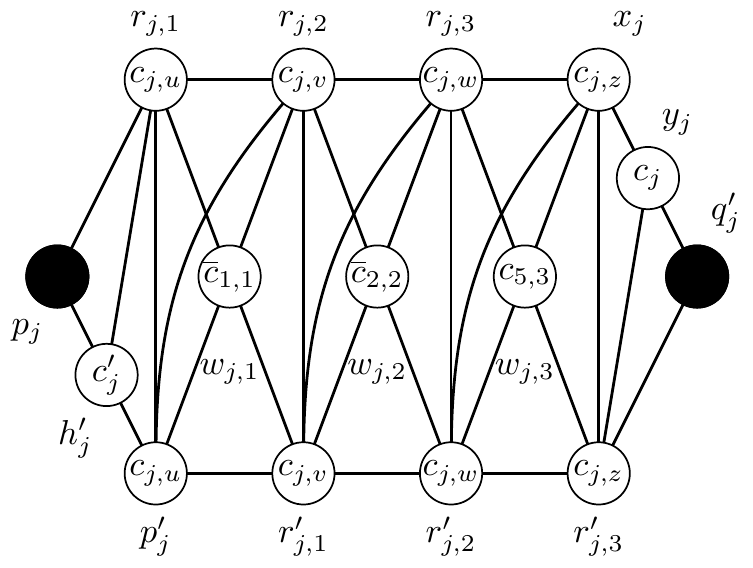}
\caption{A clause gadget $C^I_j$.}
\label{fig:chordal_clause_gadget}
\end{figure}

Finally, consider a clause gadget $C^I_j$. The addition of chords establishes additional paths between $p_j$ and $q'_j$. However, as each vertex in $\{ f_j \mid 1 \leq j < m \} \cup \{ t' \}$ is a cut vertex colored with color $c'_j$, no vertex rainbow path $R'$ from $s_0$ to $t$ can use vertex $h'_j$. It can be verified that $R'$ must still, in every $C^I_j$, use at least one of the vertices $w_{j,1}$, $w_{j,2}$, or $w_{j,3}$.
By an argument similar to Lemma~\ref{lem:g_rainbow_conn}, we have the theorem.
\end{proof}
We will then prove a stronger result for \probSrvc.
\begin{theorem}
\label{thm:srvc_interval_npc}
\probSrvc is $\NP$-complete when restricted to the class of proper interval graphs.
\end{theorem}
\begin{proof}
We assume the terminology of Theorem~\ref{thm:rvc_interval}. Given a \occsat instance $\phi$, we construct the graph $G^I_\phi$ exactly as in Theorem~\ref{thm:rvc_interval}; we will only slightly change the variable gadgets $X^I_i$ to prove our claim. Indeed, we delete the chords $(v_{i,6+2k},v_{i,18-2k})$ and add the chords $(v_{i,5+2k},v_{i,17-2k})$, where $0 \leq k \leq 2$. 

First, observe that this modification does not break the property of $G^I_\phi$ being interval. Furthermore, we can now verify $G^I_\phi$ is also claw-free. Then, consider a vertex rainbow shortest path $R$ from $a_i$ to $b_i$ after the deletion and addition of new chords. The distance $d(a_i,b_i)$ is now 10, so $R$ cannot use any of the newly added chords $(v_{i,5+2k},v_{i,17-2k})$, where $0 \leq k \leq 2$. By an argument similar to that of Theorem~\ref{thm:rvc_interval}, any $R$ must use either exactly all vertices in $N$, or all vertices $P$. Finally, any $R$ must choose from every $C^I_j$ exactly one of the vertices $w_{j,1}$, $w_{j,2}$, or $w_{j,3}$. Thus, the theorem follows.
\end{proof}
It is worth observing the modification of the variable gadget in the above theorem does not extend for \probRvc (Theorem~\ref{thm:rvc_interval}). Indeed, if the path $R$ from $a_i$ to $b_i$ is not required to be a shortest path, it is possible to construct $R$ such that a particular color from $\{ c_{i,\ell}, \overline{c}_{i,\ell} \mid 1 \leq \ell \leq 3 \}$ is avoided, possibly breaking Lemma~\ref{lem:g_rainbow_conn}.

\subsection{Cubic graphs}
\label{sec:cubic}
In this subsection, we turn our attention to regular graphs. It is easy to see that both \probRvc and \probSrvc are solvable in polynomial time on 2-regular graphs. Therefore, we will consider 3-regular graphs, i.e., cubic graphs. In contrast to previous constructions, we will need additional gadgets. Strictly speaking, the gadgets we introduce in the following are not cubic. However, when the gadgets are connected together, the resulting graph will be cubic.

Indeed, before proceeding, we will describe a parametric gadget that will serve different purposes in a construction to follow. This parametric gadget $T_k$, where $k \geq 1$, is a cycle graph of length $8k+2$. We choose two vertices $v_s$ and $v_t$ such that $d(v_s,v_t) = 4k+1$. The two $v_s$-$v_t$ paths of length $4k+1$ are broken down into $4k$ vertices $v_{i,\ell}$ and $v'_{i,\ell}$, respectively, where $i \in [k]$ and $\ell \in [4]$. The construction is finished by adding the chords $(v_{k,1},v'_{k,2})$, $(v_{k,2},v'_{k,1})$, $(v_{k,3},v'_{k,4})$, and $(v_{k,4},v'_{k,3})$, for each $k$. An example of a $T_k$ for $k=3$ is shown in Figure~\ref{fig:cubic_variable_gadget}. For each $k$, we introduce a set of three ``blocking'' colors $\{ c^*_{k,1}, c^*_{k,2}, c^*_{k,3} \}$ and color the vertices as follows: both vertices $v_{k,1}$ and $v'_{k,4}$ receive color $c^*_{k,1}$, both vertices $v_{k,3}$ and $v'_{k,1}$ receive color $c^*_{k,2}$, and both vertices $v_{k,4}$ and $v'_{k,3}$ receive color $c^*_{k,3}$. Both $v_s$ and $v_t$ receive a fresh new color that does not appear elsewhere. Exactly $2k$ vertices are now left uncolored: depending on the situation, we will color these vertices differently. However, we can still argue the following about a vertex rainbow path traversing $T_k$.

\begin{figure}
\includegraphics[scale=1]{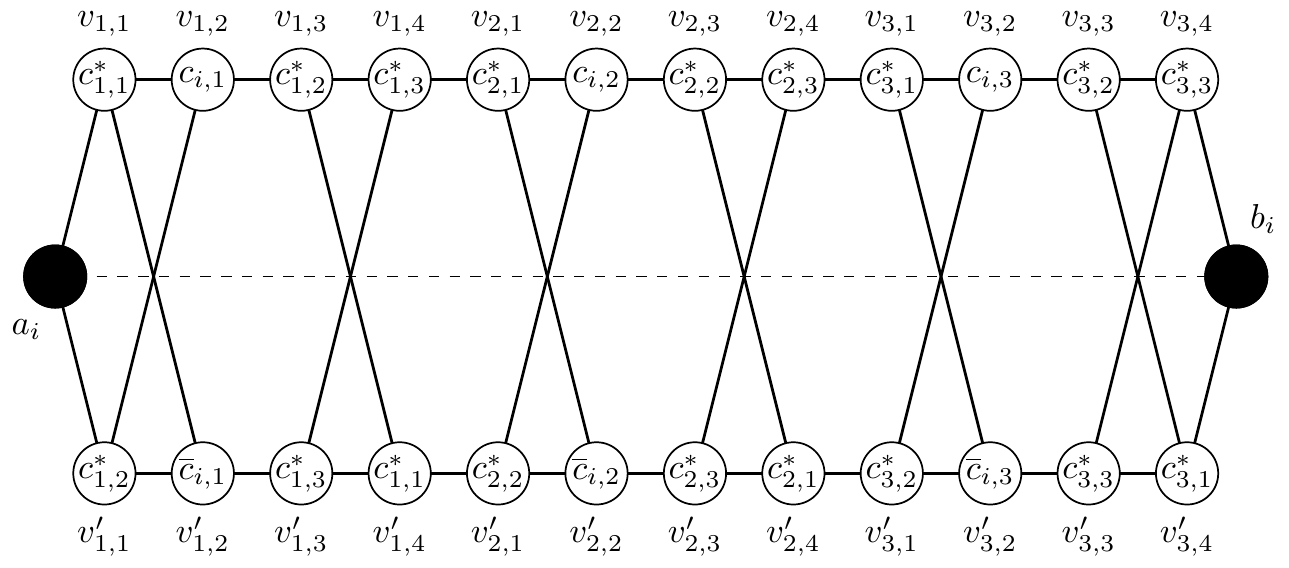}
\caption{A variable gadget $X^\Delta_i = T_3$. The vertex $v_s$ has been renamed to $a_i$, and the vertex $v_t$ to $b_i$. The dashed horizontal line divides the gadget conceptually into two segments: no vertex rainbow path from $a_i$ to $b_i$ will cross the dashed line by Lemma~\ref{lem:parametric_gadget}.}
\label{fig:cubic_variable_gadget}
\end{figure}

\begin{lemma}
\label{lem:parametric_gadget}
Let $R$ be a vertex rainbow path from $v_s$ to $v_t$ in a parametric gadget $T_k$, where $k \geq 1$. There are no $v_{\ell,i}$ and $v'_{\ell',j}$ in $R$ with $1 \leq i \leq j \leq 4$ and $\ell,\ell' \in [k]$.
\end{lemma}
\begin{proof}
For every $k$, the path $R$ must choose either $v_{k,1}$ or $v'_{k,1}$. Similarly, either $v_{k,4}$ or $v'_{k,4}$ must be chosen. If $v_{k,1}$ is chosen, then $v'_{k,4}$ cannot be chosen, as they share the same color $c^*_{k,1}$ by construction. Then, if $v'_{k,1}$ is chosen, $v'_{k,3}$ must be chosen. But then $v'_{k,3}$ and $v_{k,4}$ share the same color $c^*_{k,2}$ by construction. It follows that if $v_{k,1}$ is chosen, $v_{k,4}$ must be chosen. Symmetrically, if $v'_{k,1}$ is chosen, $v'_{k,4}$ must be chosen.
\end{proof}
The above lemma is illustrated in Figure~\ref{fig:cubic_variable_gadget}.

Informally, the color scheme described above allows us to enforce ``choose all'' type of constraints. For a clause gadget, we wish to enforce ``choose at least one'' type of constraints. Indeed, the reader should be aware that in the following, while a clause gadget is structurally a parametric gadget $T_k$, its vertex-coloring will be different. 

We will also mention that a parametric gadget $T_k$ with $k=3$ will be constructed for each variable. Here, the reader should note we do not distinguish between say vertex $v_{1,1}$ in the first variable gadget, and the vertex $v_{1,1}$ in the second variable gadget. We feel the danger for confusion is not large enough to warrant the notational burden. We are then ready to proceed with the following.
\begin{theorem}
\label{thm:rvc_cubic_npc}
\probRvc is $\NP$-complete when restricted to the class of triangle-free cubic graphs.
\end{theorem}
\begin{proof}
We assume the terminology of Theorem~\ref{thm:rvc_bip_planar}. Given a \occsat instance $\phi = \bigwedge_{j=1}^{m} c_i$ over variables $x_1,x_2,\ldots,x_n$, we follow a strategy similar to Theorem~\ref{thm:rvc_bip_planar}. We will first describe how variable and clause gadgets of a graph $G^\Delta_\phi$ are built along with their vertex-colorings.

A variable gadget $X^\Delta_i$ is a parametric gadget $T_k$, where $k=3$. To match Theorem~\ref{thm:rvc_interval} we shall rename, for each $i \in [n]$, the vertex $v_s$ to $a_i$ and the vertex $v_t$ to $b_i$ in a variable gadget $X^\Delta_i$. The uncolored vertices $v_{1,2}$, $v_{2,2}$, and $v_{3,2}$ receive colors $c_{i,1}$, $c_{i,2}$, and $c_{i,3}$, respectively. Similarly, the vertices $v'_{1,2}$, $v'_{2,2}$, and $v'_{3,2}$ receive colors $\overline{c}_{i,1}$, $\overline{c}_{i,2}$, and $\overline{c}_{i,3}$, respectively. Conceptually, these two sets of three vertices correspond to the positive and the negative $X_i$ path of Theorem~\ref{thm:rvc_bip_planar}. A variable gadget $X^\Delta_i$ along with its vertex-coloring is shown in Figure~\ref{fig:cubic_variable_gadget}.

A clause gadget $C^\Delta_j$ is a parametric gadget $T_k$, where $k=2$. To match Theorem~\ref{thm:rvc_interval} we shall rename, for each $j \in [m]$, the vertex $v_s$ to $p_j$ and the vertex $v_t$ to $q'_j$ in a clause gadget $C^\Delta_j$. We will then describe how each vertex of $C^\Delta_j$ is colored; note that we do not follow the usual coloring scheme of $T_k$ here. For convenience, let us rename $v_{1,1}$ to $r_{j,1}$, $v_{1,2}$ to $r_{j,2}$, $v_{1,3}$ to $r_{j,3}$, $v_{2,1}$ to $r_{j,4}$, and $v'_{2,3}$ to $r_{j,5}$. Also, let us rename $v'_{1,2}$ to $w_{j,1}$, $v'_{1,4}$ to $w_{j,2}$, and $v'_{2,2}$ to $w_{j,3}$. Then, the vertex $w_{j,\ell}$ for $\ell \in [3]$ is colored precisely as in Theorem~\ref{thm:rvc_bip_planar}. The vertex $v'_{1,1}$ receives color $c'_j$, and the vertex $v_{2,4}$ color $c_j$. We introduce a set of three ``blocking'' colors $\{ c^*_{j,x}, c^*_{j,y}, c^*_{j,z} \}$, and color both vertices $v_{1,4}$ and $v'_{1,3}$ with $c^*_{j,x}$, both $v_{2,2}$ and $v'_{2,1}$ with $c^*_{j,y}$, and both $v_{2,3}$ and $v'_{2,4}$ with $c^*_{j,z}$. A clause gadget along with its vertex-coloring is shown in Figure~\ref{fig:cubic_clause_head_gadgets}~(a).

\begin{figure}
\bsubfloat[]{%
  \includegraphics[scale=1]{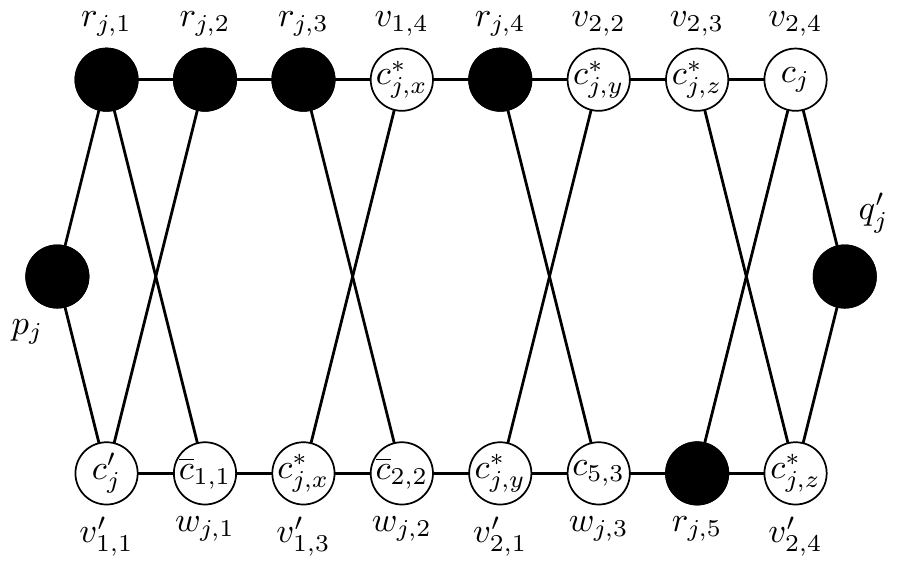}%
}
\bsubfloat[]{%
  \includegraphics[scale=1]{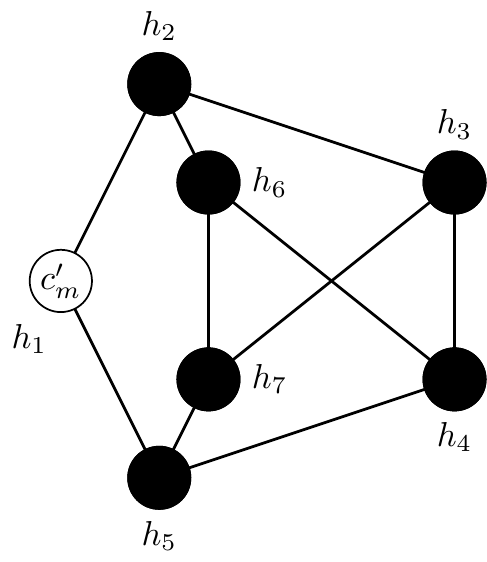}%
}
\bsubfloati\qquad\bsubfloatii
\caption{\textbf{(a)} A clause gadget $C^\Delta_j = T_2$ with some vertices renamed, and \textbf{(b)} the head gadget.}
\label{fig:cubic_clause_head_gadgets}
\end{figure}

We will then describe how variable and clause gadgets are connected together, along with some additional gadgets. After describing the additional gadgets, we will explain how they are vertex-colored. For each $1 \leq i < n$, we connect $X^\Delta_i$ with $X^\Delta_{i+1}$ by adding the edge $(b_i,a_{i+1})$. Similarly, for each $1 \leq j < m$, we connect $C^\Delta_i$ with $C^\Delta_{i+1}$ by adding the edge $(q'_j,p_{j+1})$. Let us then subdivide the edge $(q'_j,p_{j+1})$ by a new vertex $f_j$. For each $f_j$, we introduce the following \emph{dummy gadget}. A dummy gadget is constructed by starting from the cycle graph $C_5$ on the vertices $h_{j,q}$ in clockwise order, where $q \in [5]$, and two additional vertices $h_{j,6}$ and $h_{j,7}$. The construction of a dummy gadget is finished by adding the edges $(h_{j,2},h_{j,6})$, $(h_{j,3},h_{j,7})$, $(h_{j,4},h_{j,6})$, $(h_{j,5},h_{j,7})$, and $(h_{j,6},h_{j,7})$. The vertex $f_j$ is made adjacent to $h_{j,1}$ by adding the edge $(f_j,h_{j,1})$. Finally, the two components are connected by adding the edge $(b_n,p_1)$. 

We will then construct a \emph{tail gadget}, which is a parametric gadget $T_k$ with $k=m$. The vertex $v_s$ of the tail gadget is connected to $a_1$, i.e., we add the edge $(v_s,a_1)$. In addition, we construct a single dummy gadget, and connect its degree two vertex $h_1$ with $G^\Delta_\phi$ by adding the edge $(q'_m,h_1)$. For convenience, we will refer to this dummy gadget as the \emph{head gadget}. The head gadget is shown in Figure~\ref{fig:cubic_clause_head_gadgets}~(b).

We will then describe the vertex-coloring of the remaining vertices. In the tail gadget, both vertices $v_{j,2}$ and $v'_{j,2}$ receive color $c_j$, for every $j \in [m]$. Other vertices in a tail gadget follow the coloring scheme described for a $T_k$ in the beginning of Section~\ref{sec:cubic}. For each $1 \leq j < m$, we color vertex $f_j$ with color $c'_j$. Then, in the head gadget, the vertex $h_1$ receives color $c'_m$. Every other uncolored vertex of $G^\Delta_\phi$ receives a fresh new color that does not appear elsewhere. Formally, these are exactly the vertices in
\begin{equation*}
\begin{split}
Z &= \{ a_i,b_i \mid 1 \leq i \leq n \} \\
\quad &\cup \{ p_j,q'_j,r_{j,1},r_{j,2},r_{j,3},r_{j,4},r_{j,5} \mid 1 \leq j \leq m \} \\
\quad &\cup \{ h_q \mid 2 \leq q \leq 7 \} \\
\quad &\cup \{ h_{j,q} \mid 1 \leq j < m \wedge 1 \leq q \leq 7\}.
\end{split}
\end{equation*}
As each gadget is cubic and triangle-free, the graph $G^\Delta_\phi$ is cubic and triangle-free. Consider a clause gadget $C^\Delta_j$, and a vertex rainbow path $R$ traversing from $p_j$ to $q'_j$ in it. It can be observed that because $R$ cannot choose either $v'_{1,1}$ or $v_{2,4}$, it must choose at least one of the vertices $w_{j,\ell}$, where $\ell \in [3]$. Then, Lemma~\ref{lem:parametric_gadget} together with an argument similar to Lemma~\ref{lem:g_rainbow_conn} gives the theorem.
\end{proof}
Furthermore, given a positive instance $\phi$ of \occsat, it can be observed every pair of vertices is connected by a vertex rainbow shortest path, giving us the following.
\begin{theorem}
\label{thm:srvc_cubic_npc}
\probSrvc is $\NP$-complete when restricted to the class of triangle-free cubic graphs.
\end{theorem}

\subsection{$k$-regular graphs}
\label{sec:regular}
In this subsection, we show both \probRvc and \probSrvc remain $\NP$-complete on $k$-regular graphs, where $k \geq 4$. Our plan is to use the construction of Section~\ref{sec:cubic}, but add dummy vertices in a controlled manner to increase the degree of each vertex. In particular, we will need two operations detailed next.

Let $u$ and $v$ be two adjacent vertices such that $\deg(u) = \deg(v) = 3$. Let $d \geq 1$ be a constant, and consider the following \emph{degree increment} operation. We introduce a set of vertices $X = \{x_1,\ldots,x_d\}$ along with the edges $\{ (u,x), (v,x), (x,x') \mid x,x' \in X \}$. In other words, the vertices $\{u,v\} \cup X$ form a clique of size $d+2$. For each $x \in X$, we introduce a clique $W_x$ on $d+4$ new vertices $w^x_1,\ldots,w^x_{d+4}$ with an edge removed, say $(w^x_1,w^x_2) \notin W_x$. Finally, for each $x \in X$, we add the edges $(x,w^x_1)$ and $(x,w^x_2)$. We can then verify both $u$ and $v$ have degree $d+3$. Furthermore, every new vertex we added has degree $d+3$. The degree increment with $d=1$ applied to two vertices $u$ and $v$ is illustrated in Figure~\ref{fig:degree_inc_gadgets}~(a).

The degree increment operation suffices to show \probRvc is $\NP$-complete on $k$-regular graphs for $k \geq 4$. However, for \probSrvc we need to be careful not to change certain distances in our construction. For this reason, we will need an additional \emph{detour gadget} $D_{d,l}$. A building block $B$ of a detour gadget $D_{d,l}$ is the complete graph $K_{d-1}$ with two universal vertices added. The graph $B$ has $d-1$ vertices of degree $d$, and two vertices of degree $d-1$. By chaining such graphs $B$ together by adding an edge between the vertices of degree $d-1$, we obtain a detour gadget $D_{d,l}$ for which it holds that the degree of every vertex is $d$ except for two vertices that have degree $d-1$, and the diameter is $l = 2+3p$, for some $p \in \mathbb{N}^+$. A detour gadget $D_{4,5}$ is shown in Figure~\ref{fig:degree_inc_gadgets}~(b).

\begin{figure}
\bsubfloat[]{%
  \includegraphics[scale=1]{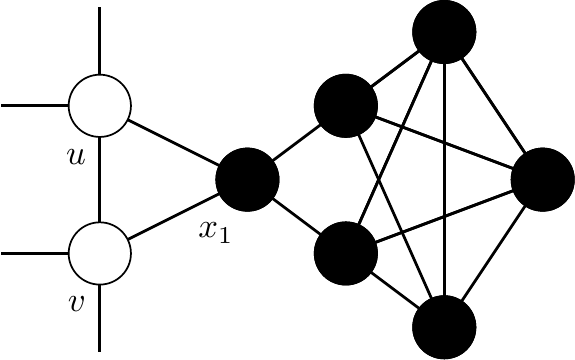}%
}
\bsubfloat[]{%
  \includegraphics[scale=1]{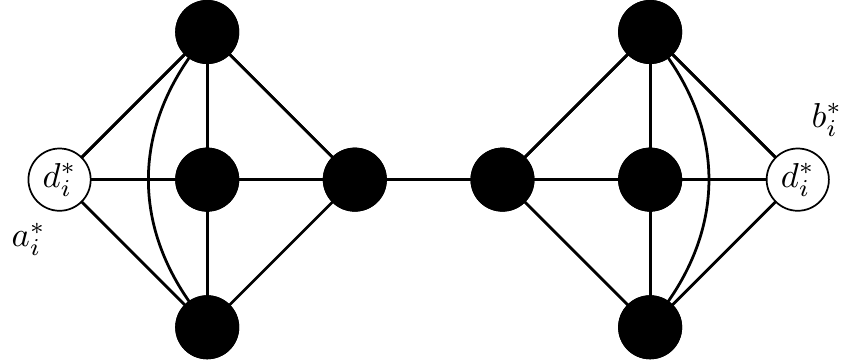}%
}
\bsubfloati\qquad\bsubfloatii
\caption{\textbf{(a)} The degree increment operation applied to $\{u,v\}$ with $d=1$, where $u$ and $v$ have degree~3. The unlabeled vertices correspond to $w^{x_1}_1,\ldots,w^{x_1}_5$. \textbf{(b)} A detour gadget $D_{4,5}$ of diameter~5.}
\label{fig:degree_inc_gadgets}
\end{figure}

We are then ready to proceed with our claim.
\begin{theorem}
\label{thm:rvc_srvc_regular_npc} 
Both \probRvc and \probSrvc are $\NP$-complete when restricted to the class of $k$-regular graphs, for every $k \geq 4$.
\end{theorem}
\begin{proof}
Consider the vertex-colored cubic graph $G^\Delta_\phi$ constructed in the proof of Theorem~\ref{thm:rvc_cubic_npc}. Through degree increment operations and addition of detour gadgets, we will transform the cubic graph $G^\Delta_\phi$ into a $k$-regular graph $G^*_\phi$, for any $k \geq 4$. Consider a variable gadget $X^\Delta_i$. We divide the vertices $v_{1,1},v_{1,2},\ldots,v_{3,4}$ into six pairs $\{v_{1,1},v_{1,2}\},\ldots,\{v_{3,3},v_{3,4}\}$. Similarly, the vertices $v'_{1,1},v'_{1,2},\ldots,v'_{3,4}$ are divided into six pairs $\{v'_{1,1},v'_{1,2}\},\ldots,\{v'_{3,3},v'_{3,4}\}$. For each of the altogether 12 pairs, we apply the degree increment operator with $d = k - 3$. We repeat this for each variable gadget in $G^\Delta_\phi$, and color each vertex arising from the operation with a fresh new color. Finally, consider the vertices $a_i$ and $b_i$ in a variable gadget $X^\Delta_i$. As the distance $d(a_i,b_i) = 13$, we introduce a detour gadget $D_{k,11}$ whose vertices of degree $k-1$ are named $a^*_i$ and $b^*_i$. By adding the edges $(a_i,a^*_i)$ and $(b_i,b^*_i)$ we ensure $d(a_i,b_i)$ remains equal to~13. For each detour gadget $D_{k,11}$, we introduce a new color $d^*_i$ that does not appear anywhere else. We color both $a^*_i$ and $b^*_i$ with color $d^*_i$. Every vertex other than $a^*_i$ and $b^*_i$ receives a fresh distinct color. This ensures that an argument similar to Lemma~\ref{lem:g_rainbow_conn} holds: no vertex rainbow path can pass through a detour gadget, as both $a^*_i$ and $b^*_i$ have the same color. For Lemma~\ref{lem:g_srvc_iff_st_path} to hold, it is enough to observe no vertex rainbow (shortest) path needs to have both $a^*_i$ and $b^*_i$ as its internal vertices. Indeed, for the remainder of the construction, each detour gadget will follow the same coloring scheme.

Let us then consider a clause gadget $C^\Delta_j$ of $G^\Delta_\phi$. Without loss, we can assume the given \occsat formula $\phi$ only contains clauses of size two and three. Indeed, clauses of size one can be removed by unit propagation. Consider the vertices $p_j$ and $q'_j$ in $C^\Delta_j$. When the corresponding clause is of size two, $d(p_j,q'_j) = 7$. Thus, similarly as above with a variable gadget, we add a detour gadget $D_{k,5}$, and connect it to $p_j$ and $q'_j$. Otherwise, the corresponding clause is of size three, and $d(p_j,q'_j) = 9$. In the obvious way, we can extend the length of the clause gadget $C^\Delta_j$ such that $d(p_j,q'_j) = 10$ by breaking the triangle-freeness of the gadget. The two vertices added to the clause gadget for this purpose receive fresh distinct colors. Then, we add a detour gadget $D_{k,8}$, and connect it with the clause gadget in the already described manner. 

Consider then a vertex $f_j$ connecting two clause gadgets, for $j \in [m-1]$. We divide $f_j$ along with its dummy gadget into four pairs of vertices, and apply the degree increment operation for each with $d=k-3$. In a similar fashion, we increase the degree of each vertex in the tail gadget, also possibly extending its length to accommodate for a detour gadget. For simplicity, we replace the head gadget as follows. We delete the vertices $h_q$ for $2 \leq q \leq 7$, and identify $h_1$ with $p_{j+1}$ of a new clause gadget $C^\Delta_{j+1}$. Each vertex of $C^\Delta_{j+1}$ receives a fresh new color (so $h_1$ still has color $c'_m$). As above, we increase the degree of each vertex in $C^\Delta_{j+1}$. 

At this point, for the obtained graph $G^*_\phi$, it holds that every vertex has degree $k$, except for two vertices $v_s$ in the tail gadget, and $q'_{j+1}$ in the clause gadget $C^\Delta_{j+1}$ replacing the head gadget. To finish the construction, we connect a detour gadget with $v_s$ and $q'_{j+1}$, extending $C^\Delta_{j+1}$ in the obvious way if necessary. This completes the proof.
\end{proof}

\section{Tractability considerations}
\label{sec:tractability_considerations}
In this section, we consider both \probRvc and \probSrvc from a structural viewpoint. We pinpoint graph classes for which both problems can be solved in polynomial time. Furthermore, we consider implications of our hardness results for parameterized algorithms, along with some positive parameterized results.

\subsection{Polynomial time solvable cases}
\label{sec:poly_cases}
A graph is said to be \emph{geodetic} if there is a unique shortest path between every pair of its vertices. It was proven by~\citet{Stemple1968} that a connected graph $G$ is geodetic if and only if every block of $G$ is geodetic. Indeed, we have the following.
\begin{observation}
\label{obs:block_geodetic}
A block graph is geodetic.
\end{observation}
This immediately leads us to the following result.
\begin{corollary}
\label{cor:srvc_block_poly}
\probSrvc is solvable in polynomial time when restricted to the class of block graphs.
\end{corollary}
More generally, a graph is said to be \emph{$k$-geodetic} if there are at most $k$ shortest paths between every pair of vertices. Quite trivially, \probSrvc is solvable in polynomial time on such graphs. This includes e.g., bigeodetic graphs~\citep{Srinivasan1988} (that is, $k=2$).

It is known that \probRc is $\NP$-complete for the class of block graphs. However, it turns out this is not the case for \probRvc. Indeed, the following lemma suggests a straightforward algorithm for the problem.
\begin{lemma}
Two distinct vertices $s$ and $t$ are rainbow vertex connected in a vertex-colored block graph if and only if each cut vertex on the unique $s$-$t$ shortest path has a distinct color.
\end{lemma}
\begin{proof}
The vertices $s$ and $t$ are rainbow vertex connected regardless of the underlying vertex coloring if $d(s,t) \leq 2$. So we can assume $d(s,t) \geq 3$. Recall that by Observation~\ref{obs:block_geodetic}, the shortest path between $s$ and $t$ is unique. We will then show that if $s$ and $t$ are rainbow vertex connected, then each cut vertex on the unique shortest $s$-$t$ path $P$ has received a different color. Suppose not, i.e., $s$ and $t$ are rainbow vertex connected, but at least two cut vertices on $P$ share the same color. But because any $s$-$t$ path uses every cut vertex on $P$, we have a contradiction. The other direction is trivial. 
\end{proof}
In other words, a vertex-colored block graph is rainbow vertex connected if and only if it is strongly rainbow vertex connected. Thus, the previous lemma establishes the following.
\begin{corollary}
\label{cor:rvc_block_poly}
\probRvc is solvable in polynomial time when restricted to the class of block graphs.
\end{corollary}

In the $st$-version of \probSrvc, the input has two additional vertices $s$ and $t$. The task is to decide whether there is a vertex rainbow shortest path between $s$ and $t$ in the graph $G$. Let us refer to this problem as \probstSrvc. We define the problem \probstRvc analogously.
\begin{lemma}
The \probstSrvc problem for cactus graphs reduces to the \probstRvc problem for cactus graphs.
\end{lemma}
\begin{proof}
Let $I = (G,\psi,s,t)$ be an instance of \probstSrvc, where $G$ is a cactus graph, and $\psi$ its vertex-coloring. In polynomial time, we will construct an instance $I' = (G',\psi',s,t)$ of \probstRvc where $G'$ is a cactus graph such that $I$ is a YES-instance of \probstSrvc if and only if $I'$ is a YES-instance of \probstRvc.

To construct $I'$, we first let $G' = G$ and $\psi' = \psi$. Then, we delete from $G'$ every vertex $w$ such that $d(s,w) + d(w,t) \neq d(s,t)$. This is achieved by running two breadth-first searches; one from $s$ and one from $t$, recording the distance to every other vertex. In other words, $G'$ contains only vertices that appear on some shortest $s$-$t$ path. Clearly, the property of being a cactus graph is closed under vertex deletion. Thus, $G'$ is a cactus graph. By observing precisely cycles of even length are preserved in $G'$, it is straightforward to verify that $I$ is a YES-instance of \probstSrvc if and only if $I'$ is a YES-instance of \probstRvc.
\end{proof}
It is shown by~\citet{Uchizawa2013} that \probstRvc can be solved in polynomial time for outerplanar graphs, which form a superclass of cacti. By applying the above reduction to each pair of vertices, we obtain the following.
\begin{corollary}
\label{cor:srvc_cactus_poly}
\probSrvc is solvable in polynomial time when restricted to the class of cactus graphs.
\end{corollary}

\subsection{Consequences for parameterized algorithms}
\label{sec:para_consequences}
It is known that both \probRc and \probRvc remain $\NP$-complete for graphs of bounded diameter. However, \probSrc is in $\XP$ parameterized by the diameter on the input graph~\citep{Lauri2015}. Indeed, by the same argument as in~\citep[Theorem~11]{Lauri2015}, we establish a similar result for the strong vertex variant.
\begin{observation}
\label{obs:srvc_diam_xp}
\probSrvc is in $\XP$ parameterized by the diameter of the input graph.
\end{observation} 
This implies \probSrvc is in $\XP$ for several other structural parameters including domination number, independence number, minimum clique cover, distance to cograph, distance to cluster, distance to co-cluster, distance to clique, and vertex cover. We refer the reader to~\citet{Komusiewicz2012} for a visualization of the relationships of many graph parameters. In fact, it can be observed the diameter of any split graph is at most three. Thus, we obtain the following for both strong variants of the problem.
\begin{corollary}
\label{cor:src_srvc_split_poly}
Both \probSrc and \probSrvc are solvable in polynomial time when restricted to the class of split graphs.
\end{corollary}  

It follows from the work of~\citet{Uchizawa2013} that \probRvc is $\NP$-complete for graphs of bounded treewidth. The pathwidth of an interval graph $G$ is $\omega(G)-1$, i.e., one less than the size of the maximum clique in $G$. We can observe the maximum clique in the graph $G^I_\phi$ constructed in Theorem~\ref{thm:rvc_interval} is of size 4. Thus, hardness of both problems for bounded pathwidth graphs follow. Furthermore, we can connect a clique of size at least 5 to $G^I_\phi$, and color each of its vertices with a fresh new color. Thus, we obtain the following.
\begin{theorem}
\label{thm:hardness_pw}
Both \probRvc and \probSrvc remain $\NP$-complete when restricted to the class of graphs with pathwidth $p$, for every $p \geq 3$.
\end{theorem}
Recall the bandwidth of a graph $G$ is one less than the maximum clique size of any proper interval supergraph of $G$, chosen to minimize its clique number. In Theorem~\ref{thm:srvc_interval_npc}, the graph constructed is already a proper interval graph. Moreover, we can verify its maximum clique size is 4. But we started the construction from the graph built in Theorem~\ref{thm:rvc_interval}. Thus, we can connect a clique of any size colored with fresh new colors to either one of the graphs, and observe the following.
\begin{theorem}
\label{thm:hardness_bw}
Both \probRvc and \probSrvc remain $\NP$-complete when restricted to the class of graphs with bandwidth $b$, for every $b \geq 3$.
\end{theorem}

Finally, one can observe Theorem~\ref{thm:hardness_pw} also implies hardness for bounded treewidth graphs. Thus, it is interesting to consider a parameter stronger than pathwidth. Indeed, tree-depth is an upper bound on the pathwidth of a graph. It was shown by~\citet{Nesetril2008} that the length of a longest path in an undirected graph $G$ is upper bounded by $2\td(G) - 2$. Using this fact in combination with the argument given in~\citep[Theorem~11]{Lauri2015}, we have the following.
\begin{observation}
\label{obs:treedepth_xp}
Both \probRvc and \probSrvc are in $\XP$ parameterized by the tree-depth of the input graph.
\end{observation}

The previous observation raises a natural question: is either problem FPT for tree-depth? Similarly, in the light of Observation~\ref{obs:srvc_diam_xp}, it is interesting to ask whether \probSrvc or \probSrc is FPT parameterized by the diameter of the input graph. In the following, we remark these questions have a positive answer.

\citet{Uchizawa2013} gave a dynamic programming algorithm for solving all four problems in $2^k n^{O(1)}$ time and exponential space, where $k$ is the number of colors used in the coloring of the input graph. Their algorithm decides whether there is a rainbow walk from an arbitrary vertex $s$ to each vertex $v \in V \setminus \{ s \}$. The crucial property is that any rainbow $s$-$v$ walk is of length at most $k$, for otherwise a color would have to repeat. We remark that their algorithm is also an FPT algorithm for any parameter that bounds the longest (shortest) path length. Indeed, if the diameter is bounded, this gives us an upper bound on the length of a walk to compute in the strong variant. The observation is similar for tree-depth.

\begin{theorem}
\label{thm:td_fpt}
All problems \probRc, \probSrc, \probRvc, and \probSrvc are $\FPT$ parameterized by the tree-depth of the input graph.
\end{theorem}

\begin{theorem}
\label{thm:diam_fpt}
Both \probSrc and \probSrvc are $\FPT$ parameterized by the diameter of the input graph.
\end{theorem}

\section{Concluding remarks}

We gave several complexity results for both \probRvc and \probSrvc (see Table~\ref{tbl:hardness_summary}). The goal was to investigate whether the complexity results for the edge variants in~\cite{Lauri2015} could be extended for the vertex variants. As the results in Table~\ref{tbl:hardness_summary} show, it is not a priori obvious how complexity is affected when considering the vertex variants for a particular graph class. This is showcased by e.g., block graphs. In the process, we obtained further negative results for the edge variants, and positive parameterized results for all four problems. 

Previously, it was shown in~\cite{Uchizawa2013} that \probRvc is $\NP$-complete for series-parallel graphs. We remark the same is true for \probSrvc. Indeed, we follow precisely the reduction given in~\cite{Uchizawa2013}, but reduce from \probSrc instead of \probRc.

From a parameterized perspective, it seems the strong variants are more tractable. Moreover, Table~\ref{tbl:hardness_summary} suggests the vertex variants are never harder than the edge variants. Is there a graph class for which say \probRvc is hard, but \probRc easy? It is also interesting to consider the complexity of the weak problem variants for split graphs. In particular, the vertex variant is trivial for split graphs of diameter~2, but what about split graphs of diameter~3?



\bibliographystyle{model1-num-names}
\bibliography{bibliography}

\section*{Appendix}
As mentioned in Subsection~\ref{sec:overview}, all of our reductions from \occsat assume each clause of the input formula has exactly three literals. For completeness, we present here clause gadgets corresponding to clauses of size two for each graph class considered.

The clause gadgets for different graph classes are shown in Figure~\ref{fig_clause_gadgets}. The first column denotes the graph class. The second column shows a clause gadget corresponding to a clause containing two literals. See the respective theorems for an explanation of the colors appearing on the vertices.

\begin{figure}[b]
\includegraphics[keepaspectratio]{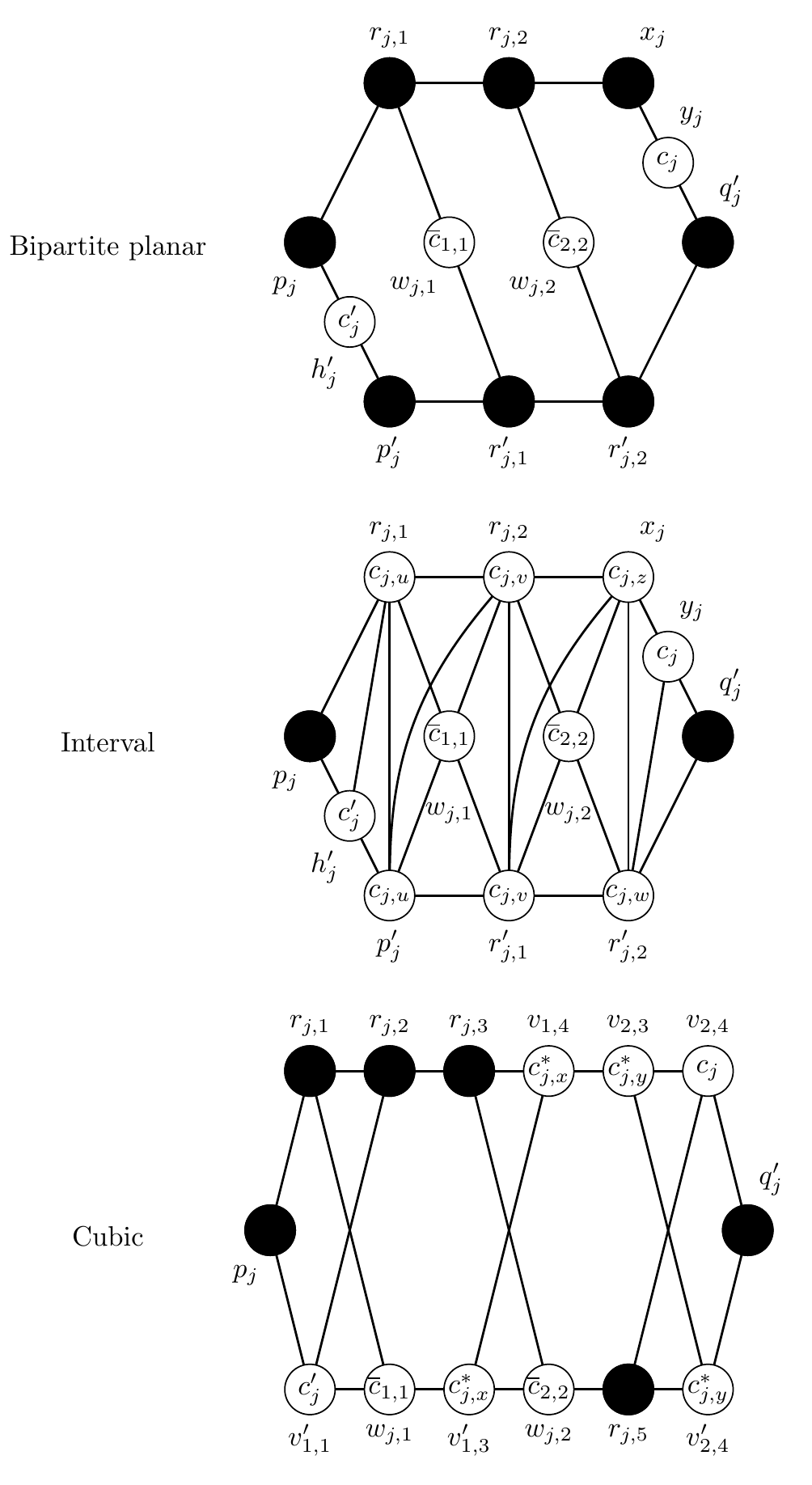}
\caption{Clause gadgets corresponding to clauses of size two for different graph classes.}
\label{fig_clause_gadgets}
\end{figure}

\end{document}